\documentclass[a4paper]{amsproc}
\usepackage{amssymb}
\usepackage{amscd}
\usepackage[dvips]{graphicx}

\theoremstyle{plain}
 \newtheorem{thm}{Theorem}[section]
 \newtheorem{prop}{Proposition}[section]
 
 \newtheorem{cor}{Corollary}[section]
\theoremstyle{definition}
 \newtheorem{exm}{Example}[section]
 \newtheorem{rem}{Remark}[section]

\numberwithin{equation}{section}

\renewcommand{\setminus}{\smallsetminus}

\def\ji {\char'032}

\def\m  {\char'176}

\font\srit=wncyi8
 \font\srrm=wncyr8

\newcommand{\R}{\mathbb{R}}

\renewcommand{\le}{\leqslant}

\renewcommand{\setminus}{\smallsetminus}

\title[NOETHER SYMMETRIES IN TIME-DEPENDENT MECHANICS]{NOETHER SYMMETRIES AND INTEGRABILITY IN TIME-DEPENDENT HAMILTONIAN MECHANICS}
\subjclass[2010]{37J15, 37J35, 37J55, 70H25, 70H33}
\author[Jovanovi\'c]{\bfseries Bo\v zidar Jovanovi\'c}

\address{
Mathematical Institute SANU \\
Serbian Academy of Sciences and Arts \\
Kneza Mihaila 36, 11000 Belgrade\\
Serbia}
\email{bozaj@mi.sanu.ac.rs}

\begin{document}
\baselineskip=14pt

\begin{abstract}
We consider Noether symmetries within Hamiltonian setting as
transformations that preserve Poincar\'e--Cartan form, i.e., as
symmetries of characteristic line bundles of nondegenerate
1-forms. In the case when the Poincar\'e--Cartan form is contact,
the explicit expression for the symmetries in the inverse Noether
theorem is given. As examples, we consider natural mechanical
systems, in particular the Kepler problem. Finally, we prove a
variant of theorem on the complete (non-commutative) integrability in terms
of Noether symmetries of time-dependent Hamiltonian systems.
\end{abstract}

\maketitle

\tableofcontents

\section{Introduction}

\subsection{}

Since Emmy Noether's paper
\cite{Noether} on integrals related to invariant variational problems, there has been a
lot of efforts on its generalization, geometrical formulation, as well as on the application in  various concrete problems (e.g, see \cite{KS, O}).
For finite dimensional Lagrangian systems, Noether's general statement \cite{Noether} takes the following simple form.


Consider a Lagrangian system $(Q,L)$, where $Q$ is a configuration
space and $L(t,q,\dot q)$ is a Lagrangian, $L: \mathbb R\times
TQ\to \R$. Let $q=(q_1,\dots,q_n)$ be local coordinates on $Q$.
The motion of the system is described by the Euler--Lagrange
equations
\begin{equation} \label{Lagrange}
\frac{d}{dt}\frac{\partial L}{\partial \dot q_i}=\frac{\partial
L}{\partial q_i}, \quad i=1,\dots,n.
\end{equation}

One of the basic principles of classical mechanics, \emph{the
Hamiltonian principle of least action}, or \emph{the principle of
stationary action}, says that the solutions of the Euler--Lagrange
equations are the critical points of the action integral
\begin{equation}\label{SL}
S_L({\gamma})=\int_a^b L(t,q,\dot q)dt
\end{equation}
in a class of curves ${\gamma}:
[a,b]\to Q$ with fixed endpoints ${\gamma}(a)=q_0$,
${\gamma}(b)=q_1$ (e.g., see \cite{Ar, MR, Jo, Ko, Wh}).

Consider the action of an one-parameter group of
diffeomorphisms $g_s$ on $\mathbb R\times Q$ with the induced
vector field
$\nu=(\tau,\xi)\vert_{(t,q)}=\frac{d}{ds}\vert_{s=0}g_s(t,q)$. After
prolongation to $\mathbb R\times TQ$, the induced vector field
reads
\begin{equation}
\label{simetrija}\hat\nu=\tau(t,q)\frac{\partial }{\partial
t}+\sum_i \xi^i(t,q)\frac{\partial}{\partial q_i}+{\hat\xi}^i(t,q,\dot
q)\frac{\partial}{\partial \dot q_i},
\end{equation}
where (e.g., see \cite{SC1, CPT})
$$
{\hat\xi}^i=\frac{\partial
\xi^i}{\partial t}-\dot q_i\frac{\partial \tau}{\partial t}+\sum_j\big(\frac{\partial \xi^i}{\partial q_j}\dot
q_j-\dot q_i\frac{\partial
\tau}{\partial q_j}\dot q_j\big).
$$

The group $g_s$ is a \emph{Noether symmetry} of the Lagrangian system if it preserves the action
functional \eqref{SL},
that is, if
\begin{equation}\label{uslov}
\frac{\partial L}{\partial t}\tau+\sum_i{\big(\frac{\partial L}{\partial
q_i}\xi^i+\frac{\partial L}{\partial \dot
q_i}{\hat\xi}^i\big)}+L \big(\frac{\partial \tau}{\partial
t}+\sum_j\frac{\partial \tau}{\partial q_j}\dot q_j \big)=0.
\end{equation}

The Noether theorem says that if $g_s$ is a Noether symmetry then
$$
I(t,q,\dot q)=\frac{\partial L}{\partial \dot q}(\xi-\tau\dot
q)+L\tau=\sum_i{\frac{\partial L}{\partial \dot q_i}(\xi^i-\tau\dot
q_i)}+L\tau
$$
is the first integral of the Euler-Lagrange equations. More
generally, if we have the invariance of \eqref{SL} modulo the
integral of ${df}/{dt}$ , that is at the right hand side of
\eqref{uslov} we have ${df}/{dt}$, for some function $f(t,q)$ (so
called \emph{gauge term}), then the integral is $I(t,q,\dot q)-f(t,q)$.

Two cases are of particular interest. If $\tau\equiv 0$ then \eqref{uslov}
reduces to the condition that the Lagrnagian $L$ is invariant with
respect to $g_s$,
$$
\sum_i{\frac{\partial L}{\partial q_i}\xi^i+\frac{\partial L}{\partial
\dot q_i}{\hat\xi}^i}=0,
$$
and the Noether integral takes the basic form
$$
I(t,q,\dot q)=\frac{\partial L}{\partial \dot
q}(\xi)=\sum_i\frac{\partial L}{\partial \dot q_i}\xi^i.
$$
In particular, when the Lagrangian does not depend on $q_1$, then $q_1$ is \emph{ignorable (cyclic) coordinate}
and the integral is ${\partial L}/{\partial \dot q_1}$ (e.g, see \cite{Ar, Wh}).

Secondly, if the Lagrangian does not depend on time, we can take
the translations in time: $g_s(t,q)=(t+s,q)$. Then the vector
field \eqref{simetrija} is simply ${\partial}/{\partial t}$ and
the integral $I$ becomes the energy of the system multiplied by
$-1$:
$$
I=-E=-\sum_i{\frac{\partial L}{\partial \dot q_i}\dot q_i}+L.
$$

\subsection{}
The Noether theorem can be seen as a part of time-dependent
mechanics that is studied and geometrically formulated both in the
Lagrangian and Hamiltonian setting (e.g., see \cite{Al, BEADM,
CLL, CMN, Cr, CPT, D, DV, LMP, LM, KKPS, GS3, GS2} and references
therein\footnote{This list is far away to be a compete list of
contributions on the subject.}). In Sarlet and Cantrijn
\cite{SC1}, one can find a review of various approaches on the
Noether theorem in the Lagrangian framework for velocity dependent
transformations, as well as a geometrical setting for the
equivalence of the first integrals and symmetries of the
Lagrangian system considered as a characteristic system of the
two-form $d\alpha$ ($\alpha$ being
Poincar\'e--Cartan form). 

The aim of this paper is to present the problem through the perspective of contact geometry, continuing the study of the Maupertuis principle,
isoenergetic, and partial integrability \cite{Jo, Jo2}.
We consider Noether symmetries
as symmetries of characteristic line bundles of nondegenerate
1-forms (Theorems \ref{prva}, \ref{druga}).
In the case of time-dependent Hamiltonian systems, Noether symmetries are transformations that preserve Poincar\'e--Cartan
form (see Proposition \ref{stav}), and, via
Legendre transformation, this is equivalent to Crampin's notion of symmetry of Lagrangian systems \cite{Cr}.
This will allow us to use contact geometry for the inverse Noether theorem in Section 4.

The notion of a weak Noether
symmetry is also given and the relation with the Noether
symmetries is established (Proposition \ref{jaka-slaba}). The Noether symmetry
is a natural generalization of classical Noether symmetry
described above (see Proposition \ref{uopstenje}), while the notion of the week
Noether symmetry corresponds to the classical Noether symmetry
with the gauge term.

In the case when the Poincar\'e--Cartan form is contact, the explicit
expression for the Noether symmetry for a given first integral
without using the gauge term is given (Theorem \ref{treca}). As examples, we
consider natural mechanical systems (Corollary \ref{posledica}), in particular the Kepler
problem (Example \ref{Kepler}).

Finally, in Section 5, we obtain a variant of the complete
(non-commutative) integrability in terms of Noether symmetries of
time-dependent Hamiltonian systems (Theorem \ref{1.i}, Corollary
\ref{2.i}).

\section{Noether symmetries in the Hamiltonian formulation}

\subsection{} Let $(Q,L)$ be a Lagrangian system.
The \emph{Legendre transformation} $\mathbb FL: TQ\to T^*Q $ is defined
by
\begin{equation}\label{legendre}
\mathbb FL(t,q,\xi)\cdot \eta=
\frac{d}{ds}\vert_{s=0}L(t,q,\xi+s\eta) \quad \Longleftrightarrow
\quad p_i=\frac{\partial L}{\partial \dot q_i}, \quad i=1,\dots,
n,
\end{equation}
where $\xi,\eta\in T_q Q$ and $(q_1,\dots,q_n, p_1,\dots,p_n)$ are
canonical coordinates of the cotangent bundle $T^*Q$. In order to
have a Hamiltonian description of the dynamics we suppose that the
Legendre transformation \eqref{legendre} is a diffeomorphism. The
corresponding Lagrangian $L$ is called \emph{hyperregular}
\cite{MR}.

Let $L(t,q,\dot q)$ be a hyperregular Lagrangian. We can pass from
velocities $\dot q_i$ to the momenta $p_j$ by using the standard
Legendre transformation \eqref{legendre}. In the coordinates
$(q,p)$ of the cotangent bundle $T^*Q$, the equations of motion
\eqref{Lagrange} read:
\begin{equation} \label{1}
\frac{dq_i}{dt}=\frac{\partial H}{\partial p_i},\qquad
\frac{dp_i}{dt}=-\frac{\partial H}{\partial q_i}, \qquad
i=1,\dots,n,
\end{equation}
where the Hamiltonian function $H(t,q,p)$ is the \emph{Legendre
transformation} of  $L$
\begin{equation} \label{haml}
 H(t,q,p)=E(t,q,\dot q)\vert_{\dot q=\mathbb FL^{-1}(t,q,p)}=\mathbb FL(t,q,\dot q)\cdot \dot q-L(t,q,\dot
q)\vert_{\dot q=\mathbb FL^{-1}(t,q,p)}.
\end{equation}

Let $p dq=\sum_i p_i dq_i$ be the \emph{canonical 1-form}
and
$$
\omega=d(p dq)=dp \wedge dq=\sum_{i=1}^n dp_i\wedge dq_i
$$
the {\it canonical symplectic} form of the cotangent bundle
$T^*Q$. The equations \eqref{1} are Hamiltonian, i.e.,  they can be
written as $\dot x=X_H$, where the Hamiltonian vector field $X_H$
is defined by
\begin{equation*}\label{Hamiltonian}
i_{X_H}\omega (\,\cdot\,)=\omega(X_H,\,\cdot\,)=-dH(\,\cdot\,).
\end{equation*}

\subsection{Noether symmetries} Consider the \emph{Poincar\'e--Cartan} 1-form
$$
\alpha=pdq-Hdt
$$
on the extended phase space $\R\times T^*Q (t,q,p)$, where $H:
\R\times T^*Q  \to \R$ is a Hamiltonian function. The phase
trajectories of the canonical equations \eqref{1} are extremals of
the action
\begin{equation}\label{poincare}
A_H(\gamma)=\int_\gamma \alpha=\int_\gamma pdq-Hdt
\end{equation}
in the class of curves $\gamma(t)=(t,q(t),p(t))$ connecting the
subspaces $\{t_0\} \times T^*_{q_0} Q$ and $\{t_1\}\times
T^*_{q_1} Q$ \cite{Ar, MR} (Poincar\'e's modification of the
Hamiltonian principle of least action \cite{P}). Obviously, we can
replace $(T^*Q,dp\wedge dq)$ by an arbitrary exact symplectic
manifold $(P,\omega=d\theta)$.

We shall say that the vector field
$$
\zeta=\tau(t,q,p)\frac{\partial
}{\partial t}+\sum_i{\xi^i(t,q,p)\frac{\partial}{\partial
q_i}+\eta^i(t,q,p)\frac{\partial}{\partial p_i}},
$$
i.e., the induced one-parameter group of diffeomeomorphisms
$g_s^\zeta$ of $\mathbb R\times T^*Q$,
$$
\zeta=\frac{d}{ds}g_s^\zeta\vert_{s=0}(t,q,p),
$$
is a {\it Noether symmetry} of the Hamiltonian system \eqref{1} if
the {Poincar\'e--Cartan} 1-form is preserved.
Then, by the
analogy with the Lagrangian formulation, $g_s^\zeta$
preserves the action functional \eqref{poincare}. The above
definition is suitable for a contact approach presented in Section 4.

We shall say that $\zeta$ is a \emph{weak Noether symmetry} if we have the invariance of the perturbation of the
{Poincar\'e--Cartan} 1-form 1-form by  a closed 1-form $\beta$ modulo the differential of the function $f$:
\begin{equation}\label{wns}
L_{\zeta}(pdq-Hdt+\beta)=df.
\end{equation}

The function $f(t,q,p)$ plays a role of a gauge term
in the classical formulation. The closed form $\beta$ corresponds to the fact that the
solutions of the canonical equation \eqref{1} are also extremal of the perturbed action
$$
A_{H}(\gamma)=\int_\gamma p dq-Hdt+\beta.
$$

\begin{thm}\label{prva}
Let $\zeta$ be a waek Noether symmetry satisfying \eqref{wns}. Then
\begin{itemize}

\item[i)  ] The function
$$
J=i_\zeta(p dq-H\,dt+\beta)-f=\sum_i{p_i\xi^i}-H\tau+\beta(\zeta)-f
$$
is the first integral of the Hamiltonian equations \eqref{1}.

\item[ii) ] The integral $J$ is also preserved under the flow of
$g_s^\zeta$.

\item[iii)] The one-parameter group of diffeomorphisms $g_s^\zeta$
permutes the trajectories of the Hamiltonian equations in the
extended phase $\mathbb R\times T^*Q$ space modulo
reparametrization.

\end{itemize}
\end{thm}

The notion of week Noether symmetries for $\beta=0$ is equivalent, via
Legendre transformation, to the symmetries of Lagrangin systems
considered by Crampin \cite{Cr}, see also  Sarlet and Cantrijn
\cite{SC1, SC2}. The proof of Theorem \ref{prva} is similar to the
proofs presented in \cite{SC1, SC2, Cr} and for the completeness of the exposition it will be given in the
next section (see the proof of Theorem \ref{druga}).
Also, recently, a similar approach to the higher order Lagrangian problems is given in \cite{FS}.

By definition, $\zeta$ is a Noether symmetry if and only if the Lie
derivative of the {Poincar\'e--Cartan} 1-form vanish:
\begin{align*}
0 &= L_\zeta\alpha\\
&= i_\zeta(d\alpha)+d(i_\zeta\alpha)\\
&= i_\zeta(dp\wedge dq-dH\wedge dt)+ d(p\xi-H\tau)\\
&= \eta dq-\xi dp + \tau dH-dH(\zeta)dt+ d(p\xi-H\tau)\\
&= \eta dq + pd\xi-Hd\tau-dH(\zeta)dt.
\end{align*}

Comparing the components with $dp_j$, $dq_j$, and $dt$ we get the
following statement.

\begin{prop}\label{stav}
$\zeta$ is a Noether symmetry if and only if
\begin{eqnarray}
\label{u1} && \sum_i p_i\frac{\partial\xi^i}{\partial
p_j}-H\frac{\partial\tau}{\partial p_j}=0, \\
\label{u2} && \eta^j+\sum_i p_i\frac{\partial\xi^i}{\partial
q_j}-H\frac{\partial\tau}{\partial q_j}=0, \qquad \qquad\quad j=1,\dots,n,\\
\label{u3} && \sum_i \big(p_i\frac{\partial\xi^i}{\partial
t}-\eta^i\frac{\partial
H}{\partial p_i}-\xi^i\frac{\partial H}{\partial
q_i}\big)-H\frac{\partial\tau}{\partial t}-\tau\frac{\partial H}{\partial t}=0.
\end{eqnarray}
\end{prop}

\begin{prop}\label{uopstenje}
If the Lagrangian $L$ and the Hamiltonian $H$ are
related by the Legendre transformation \eqref{legendre},
\eqref{haml} and if $\zeta$ is a Noether symmetry of the
Hamiltonian equation \eqref{1} with $\xi^i=\xi^i(t,q)$,
$\tau=\tau(t,q)$, then the classical invariance condition \eqref{uslov} is
satisfied.
\end{prop}

\begin{proof} Since $\xi^i$ and $\tau$ do not depend on $p$, the
conditions \eqref{u1} is satisfied. By expressing $\eta^i$ from
\eqref{u2} and substituting into \eqref{u3} we get the following
equation
\begin{align}\label{u4}
\sum_j p_j\big(\frac{\partial\xi^j}{\partial t}+\sum_i\frac{\partial\xi^j}{\partial q_i}\frac{\partial H}{\partial
p_i}\big)-& H\big(\frac{\partial\tau}{\partial t}+\sum_j \frac{\partial\tau}{\partial q_j}\frac{\partial H}{\partial p_j}\big)\\
&-\sum_j\xi^j\frac{\partial H}{\partial q_j}-\tau\frac{\partial H}{\partial t}=0.\nonumber
\end{align}
On the other hand, \eqref{uslov} transforms to
\begin{align}
\label{u5} \frac{\partial L}{\partial t}\tau+\sum_i\frac{\partial
L}{\partial q_i}\xi^i+&\sum_i\frac{\partial L}{\partial \dot
q_i}\big(\frac{\partial \xi^i}{\partial t}+\sum_j\frac{\partial
\xi^i}{\partial q_j}\dot q_j\big)\\
&-\big(\frac{\partial
\tau}{\partial t}+\sum_i\frac{\partial \tau}{\partial q_i}\dot
q_i\big)\big(\sum_j\frac{\partial L}{\partial \dot q_j}\dot
q_j-L\big)=0.\nonumber
\end{align}

Since the Legendre transformation \eqref{legendre}, \eqref{haml}
implies well known identities
$$
\frac{\partial L}{\partial t}=-\frac{\partial H}{\partial t},\qquad \frac{\partial L}{\partial q_i}=-\frac{\partial H}{\partial q_i},
\qquad \dot q_i=\frac{\partial H}{\partial p_i},
$$
the equations \eqref{u4} and \eqref{u5} are equivalent.
\end{proof}

Therefore,
we can consider the above definition of a Noether symmetry as a
natural generalization of the classical one.

\section{Noether symmetries of characteristic line bundles}

\subsection{} Let $(M,\alpha)$ be a $(2n+1)$--dimensional manifold
endowed with a \emph{nondegenerate} 1-form $\alpha$. This mean that
$d\alpha$ has the maximal rank $2n$. The kernel of
$d\alpha$ defines one dimensional distribution
$$
\mathcal L=\bigcup_x \mathcal L_x, \qquad \mathcal L_x=\ker d\alpha\vert_x
$$
 of the
tangent bundle $TM$ called \emph{characteristic line bundle}. Also,
at every point $x\in M$ we have the \emph{horizontal space} $\mathcal H_x$ defined
by
$$
\mathcal H_x=\ker\alpha\vert_x.
$$

In the case when $\alpha\ne 0$, $d\alpha\ne 0$ on $M$, then the
collection of horizontal subspaces $ \mathcal H=\bigcup_x \mathcal
H_x=\bigcup_x\ker\alpha\vert_x $ is a nonintegrable distribution
of $TM$, called \emph{horizontal distribution}. If, in addition,
$\alpha\wedge(d\alpha)^n\ne 0$, then $\alpha$ is a \emph{contact form} and  $(M,\alpha)$ is a
\emph{strictly contact manifold} \cite{LM}.  The horizontal distribution $\mathcal H$
is also referred as \emph{contact} distribution.

The following variational statement is well known.

\begin{thm} \label{var} The integral curves $\gamma: [a,b]\to M$ of the
characteristic line bundle $\mathcal L$ are extremals of the
action functional
$$
A(\gamma)=\int_\gamma \alpha=\int_a^b \alpha(\dot\gamma)dt
$$
in the class of variations $\gamma_s(t)$, such that
$\delta\gamma(a)$ and $\delta\gamma(b)$ are horizontal vectors.
\end{thm}

Here,  a {\it variation} of a curve $\gamma: [a,b]\to M$ is a
mapping: $\Gamma: [a,b]\times [0,\epsilon] \to M$, such that
$\gamma(t)=\Gamma(t,0)$, $t\in [a,b]$,
$\delta\gamma(t)=\frac{\partial \Gamma}{\partial s}\vert_{s=0}\in
T_{\gamma(t)} M$, and $\gamma_s(t)=\Gamma(t,s)$.

\begin{figure}[ht]
\includegraphics[width=90mm, height=60mm]{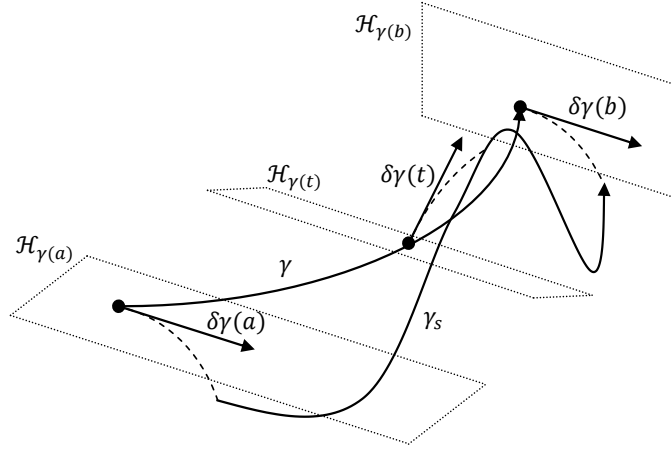}
\caption{Variation of a curve $\gamma$, such that $\delta\gamma(a)$
and $\delta\gamma(b)$ are horizontal vectors.}
\end{figure}

The proof is a direct consequence of Cartan's  formula (e.g., see
\cite{Gr})
$$
L_{\partial/\partial
s}\Gamma^*\alpha\vert_{(t,0)}=\gamma^*(i_{\delta\gamma(t)}
d\alpha)+d\gamma^*(\alpha(\delta\gamma(t))),
$$
which implies
\begin{equation}\label{izvod}
\frac{d}{ds}\left(\int_{\gamma_s}
\alpha\right)\bigg\vert_{s=0}=\int_a^b
d\alpha(\delta\gamma(t),\dot\gamma(t))dt+\alpha(\delta\gamma(b))-\alpha(\delta\gamma(a)).
\end{equation}
For $\delta\gamma(a)$ and $\delta\gamma(b)$ being horizontal,
the expression \eqref{izvod} is equal to zero if and
only if $\dot \gamma$ is in the kernel of the form $d\alpha$. That
is, $\gamma(t)$ is an integral curve of the line bundle $\mathcal
L$.

\begin{exm} \label{primer} As an example we can take the extended phase space
endowed with the {Poincar\'e--Cartan} 1-form
\begin{equation}\label{cm}
(\R\times T^*Q (t,q,p),pdq-Hdt)
\end{equation}
(e.g., see \cite{Ar}). The sections of $\ker d(pdq-Hdt)$ are of the form
$$
Z_\mu=\mu Z,
$$
where
\begin{equation}\label{reb1}
Z=\frac{\partial}{\partial t}+\sum_i{\frac{\partial H}{\partial
p_i}\frac{\partial}{\partial q_i}- \frac{\partial H}{\partial
q_i}\frac{\partial}{\partial p_i}}.
\end{equation}
and $\mu=\mu(t,q,p)$ are smooth functions. Therefore, in this
case, Theorem \ref{var} implies Hamiltonian principle of least
action in the extended phase space. Here, the vector space $T^*_q
Q$, considered as a subspace of $T_{(t,q,p)}\R\times T^*Q$, is a
subspace of the horizontal space
\begin{equation}\label{HS}
\mathcal H_{(t,q,p)}=\big\{\tau\frac{\partial
}{\partial t}+\sum_i{\xi^i\frac{\partial}{\partial
q_i}+\eta^i\frac{\partial}{\partial p_i}}\,\,\big\vert\,\, \tau,\xi^i,\eta^i\in\R,\,\, \sum_i p_i\xi^i=\tau H(t,q,p)\big\}.
\end{equation}
\end{exm}

\begin{rem}{\rm
The vector field \eqref{reb1} is determined by the conditions $i_Z
(dp\wedge dq -dH\wedge dt)=0$ and $dt(Z)=1$. In other words, $Z$
can be seen also as the Reeb vector field of the cosymplectic
manifold $ (\R\times T^*Q (t,q,p),dp\wedge dq -dH\wedge dt,dt).
$
Recall that a cosymplectic manifold $(M,\omega,\eta)$ is a
$(2n+1)$--dimensional manifold $M$ endowed with a closed 2-form
$\omega$ and a closed 1-form $\eta$, such that
$\eta\wedge\omega^n$ is a volume form, which is a natural
framework for the time-dependent Hamiltonian mechanics  (see
\cite{Al, CLL, MNY}). }\end{rem}

\subsection{Noether symmetries and integrals}
Consider the equation
\begin{equation}\label{clb}
\dot x=Z,
\end{equation}
where $Z$ is a section of $\mathcal L$ ($i_Z d\alpha=0$).

We shall say that the vector field $\zeta$, i.e., the induced
one-parameter group of diffeomeomorphisms
$g_s^\zeta$,
is a \emph{Noether symmetry} of the equation \eqref{clb} if it
preserves the 1-form $\alpha$. A similar definition for exterior
differential systems is given in \cite{Gr}.

Note that the vector field $Z$ is also a section of the characteristic
line bundle of a nondegenerate 1-form $\alpha+\beta$, where $\beta$ is
arbitrary closed 1-form on $M$. We refer to the vector field $\zeta$ as a \emph{weak Noether symmetry} if we have the invariance of the perturbation of the
 1-form $\alpha$ by a closed 1-form $\beta$ modulo the differential of the function $f$:
\begin{equation}\label{wns1}
L_{\zeta}(\alpha+\beta)=df.
\end{equation}

\begin{thm}\label{druga}
Let $\zeta$ be a weak Noether symmetry that satisfies \eqref{wns1}. Then:

\begin{itemize}

  \item[i)  ]  The function $ J=i_\zeta(\alpha+\beta)-f$ is the first integral of
\eqref{clb}.

  \item[ii) ] The integral $J$ is preserved under the flow of
  $g_s^\zeta$ as well:
$dJ(\zeta)=0$.

  \item[iii)]  The commutator of vector fields $[Z,\zeta]$ is a
section of $\mathcal L$, i.e., $g_s^\zeta$ permutes the
trajectories of \eqref{clb} modulo reparametrization.

\end{itemize}
\end{thm}

\begin{figure}[ht]
\includegraphics[width=95mm, height=55mm]{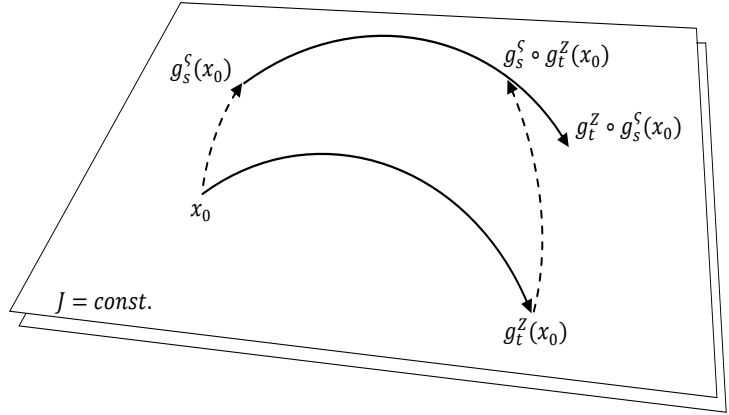}
\caption{Illustration of Theorem \ref{druga}.}
\end{figure}

\begin{proof}
i) From the definition \eqref{wns1} and Cartan's formula
$L_\zeta=i_\zeta\circ d+d\circ i_\zeta$, we have
\begin{equation}\label{izvod2}
i_\zeta d\alpha=-d(i_\zeta(\alpha+\beta))+df=-dJ,
\end{equation}
which proves the first assertion of the statement:
\begin{equation}\label{ZJ}
Z(J)=i_Z(-i_\zeta d\alpha)=d\alpha(Z,\zeta)=0.
\end{equation}

Alternatively, we have the variational interpretation of the first integral. Let $\gamma: [a,b]\to M$ be the trajectory of \eqref{clb} and consider
the variation $\gamma_s=g^\zeta_s(\gamma)$, $\delta\gamma(t)=\zeta\vert_{\gamma(t)}$. From \eqref{wns1}, \eqref{izvod} we get, respectively,
\begin{align*}\label{izvod2}
\frac{d}{ds}\left(\int_{\gamma_s}
\alpha+\beta\right)\bigg\vert_{s=0}&=\int_{\gamma_s} df=f(b)-f(a),\\
\frac{d}{ds}\left(\int_{\gamma_s}
\alpha+\beta\right)\bigg\vert_{s=0}&=\int_a^b
d(\alpha+\beta)(\zeta\vert_{\gamma(t)},\dot\gamma(t))dt+(\alpha+\beta)(\zeta\vert_b)-(\alpha+\beta)(\zeta\vert_a)\\
&=(\alpha+\beta)(\zeta\vert_b)-(\alpha+\beta)(\zeta\vert_a).
\end{align*}
Therefore,
$(\alpha+\beta)(\zeta\vert_a)-f(a)=(\alpha+\beta)(\zeta\vert_b)-f(b)$.

\

ii) Similarly as the equation \eqref{ZJ}, \eqref{izvod2} implies
$$
\zeta(J)=i_\zeta(-i_\zeta d\alpha)=d\alpha(\zeta,\zeta)=0.
$$

iii) We need to prove that $[\zeta,Z]$ belongs to the kernel of $d\alpha$. We have
\begin{align*}
  i_{[\zeta,Z]}d\alpha =& L_\zeta(i_Z d\alpha)-i_Z(L_\zeta d\alpha) \\
   =& -i_Z(i_\zeta d^2\alpha+d(i_\zeta(d\alpha))) \\
   =& -i_Zd(dJ) = 0.
\end{align*}
\end{proof}

\subsection{}
It is clear that in the study of integrals of the Hamiltonian
equations \eqref{1}, we can consider arbitrary section $Z_\mu=\mu
Z$ of $\mathcal L$, where $\mu(t,q,p)$ is a function that is
almost everywhere different from zero. The normalization
$dt(Z_\mu)=1$ implies $\mu\equiv 1$. In the case when $Z$ is
transversal to the horizontal spaces \eqref{HS}:
\begin{equation}\label{kontaktni uslov}
\rho=i_Z(pdq-Hdt)=p\frac{\partial H}{\partial p}-H\ne 0,
\end{equation}
there is another natural normalization $(pdq-Hdt)(Z_\mu)=1$,
$\mu=\rho^{-1}$,
\begin{equation}\label{reb}
Z_{\rho^{-1}}=\rho^{-1}\frac{\partial}{\partial
t}+\rho^{-1}\sum_i\big(\frac{\partial H}{\partial
p_i}\frac{\partial}{\partial q^i}- \frac{\partial H}{\partial
q_i}\frac{\partial}{\partial p_i}\big).
\end{equation}

The condition \eqref{kontaktni uslov} is equivalent to the
property that \eqref{cm} is a strictly contact manifold with the
Reeb vector field \eqref{reb} (see \cite{LM}). If the Hamiltonian
$H$ is obtained from the Lagrangian $L$ under  the Legendre
transformation \eqref{legendre}, \eqref{haml}, the function $\rho$
is the Lagrangian $L(t,q,\dot q)\vert_{\dot q=\dot q(t,q,p)}$. In
\cite{LM} it is referred as an {\it elementary action}.

\section{Inverse Noether theorem}

A natural question is the converse of the Noether theorem (e.g., see \cite{SC1}): if $F$
is the integral of \eqref{1}, is there a
Noether symmetry $\zeta$, such that $F=i_\zeta\alpha$?

A geometrical setting
for the equivalence of the first integrals
and week symmetries can be found in \cite{Cr, SC1}.
With the above notation, one should firstly construct a vector field $\zeta$, such that $i_\zeta d\alpha=dF$. Then $\zeta$ is a week Noether symmetry
with
$L_\zeta\alpha=df$, where $f=F+i_\zeta\alpha$.

It appears that the contact approach provides a simple explicit expression for the Noether symmetry for a generic Hamiltonian function.

\subsection{Contact Hamiltonian vector fields}
Let $(M,\alpha)$ be a {strictly contact manifold}. Then the
contact distribution $\mathcal H$ is transversal to the
characteristic line bundle $\mathcal L$:
$$
TM=\mathcal L \oplus \mathcal H.
$$

A vector field $X$ that
preserves $\mathcal H$:
$$
(g^X_t)_*\mathcal H=\mathcal H \quad \Longleftrightarrow
  \quad L_X\alpha=\lambda\alpha, \quad \text{for some smooth function }\lambda
$$
is called \emph{contact vector field}.
There is a distinguish contact vector field,
the \emph{Reeb vector field} $Z$, uniquely defined
by
\begin{equation}\label{REEB}
i_Z\alpha=1, \qquad i_Z d\alpha=0.
\end{equation}

For a given function $f$, one can associate the \emph{contact  vector field $Y_f$ with Hamiltonian} $f$:
\begin{equation*}
Y_f = f Z+\hat{Y}_f,
\label{iso}
\end{equation*}
where $\hat{Y}_f$ is a horizontal vector field defined by
$i_{\hat{Y}_f}d\alpha=-\big(df-Z(f)\alpha\big)$ (e.g., see \cite{LM}).
The mapping $f\mapsto Y_f$ is a bijection between smooth functions and contact vector fields on $M$.
The inverse mapping is simply the contraction: $f=i_{Y_f}\alpha$.
In particular, the Hamiltonian of the Reeb vector field is $f\equiv 1$.

Note that
$$
\mathcal L_{Y_f}\alpha=Z(f)\alpha,
$$
i.e., for the Reeb flow we have the inverse Noether theorem directly:
$Y_f$ is a Noether symmetry of the Reeb flow if and only if $f$ is the integral of the Reeb flow.

\subsection{Inverse Noether theorem}
If the elementary action is different from zero \eqref{kontaktni uslov}, a Noether symmetry $\zeta$ of the
Hamiltonian equation \eqref{1} is a contact vector field of \eqref{cm} with the Hamiltonian function $J=i_\zeta(pdq-Hdt)$.

We say that $H$ is a \emph{generic Hamiltonian}, if the condition \eqref{kontaktni uslov} hold for an open dense subset $U_H$ of  $\R\times T^*Q$.
From now one, we assume that $H$ is generic. Thus, we have:

\begin{thm}\label{treca}
To every integral $F$ of the Hamiltonian equation \eqref{1}, we can associate unique Noether symmetry $\zeta$ on $U_H$,  such that the corresponding Noether integral is equal to $F$: $i_\zeta(pdq-Hdt)=F$. The vector field $\zeta$ reads:
\begin{equation*}
\zeta=\tau(t,q,p)\frac{\partial
}{\partial t}+\sum_i{\xi^i(t,q,p)\frac{\partial}{\partial
q_i}+\eta^i(t,q,p)\frac{\partial}{\partial p_i}},
\end{equation*}
where the coefficient $\tau,\xi^i,\eta^i$
are given by
\begin{align}
\label{a1} \tau =& \rho^{-1}F -\rho^{-1}\sum_j \frac{\partial F}{\partial p_j}p_j, \\
\label{b1} \xi^i =& \rho^{-1}F\frac{\partial H}{\partial p_i}-\rho^{-1}\sum_j \frac{\partial F}{\partial p_j}p_j\frac{\partial H}{\partial p_i}+\frac{\partial F}{\partial p_i},\\
\label{c1} \eta^i =& -\rho^{-1}F\frac{\partial H}{\partial q^i}+\rho^{-1}\sum_j\frac{\partial F}{\partial p_j}p_j\frac{\partial H}{\partial q_i}-\frac{\partial F}{\partial q_i}, \qquad i=1,\dots,n.
\end{align}
In particular, if the invariant regular hypersurface $M_c=\{F(t,q,p)=c\}$ is a subset of $U_H$, the vector field $\zeta$ is well defined on the whole $M_c$.
\end{thm}

\begin{proof} The required vector field $\zeta$ is
the contact vector field with the Hamiltonian $F$:
$$
\zeta=Y_F=F Z_{\rho^{-1}}+\hat{Y}_F,
$$
where the Reeb vector field $Z_{\rho^{-1}}$ is given by
\eqref{reb} and the the coefficient $a,b^i,c^i$ of the horizontal
vector field
$$
\hat{Y}_F=a(t,q,p)\frac{\partial
}{\partial t}+\sum_i b^i(t,q,p)\frac{\partial}{\partial
q_i}+c^i(t,q,p)\frac{\partial}{\partial p_i},
$$
are uniquely determined by the conditions:
\begin{eqnarray}
\label{horizontal}
&& i_{\hat{Y}_F}(pdq-Hdt)=\sum_i{b^ip_i}-Ha=0,\\
&& \label{horizontal2} i_{\hat{Y}_F}(dp\wedge dq-dH\wedge
dt)=-\big(dF-Z_{\rho^{-1}}(F)(pdq-Hdt)\big)=-dF.
\end{eqnarray}
Here we used the fact that  $F$ is the integral of the Hamiltonian equations \eqref{1}:
$$
Z_{\rho^{-1}}(F)=\rho^{-1}\big(\frac{\partial F}{\partial
t}+\sum_i\frac{\partial F}{\partial q_i}\frac{\partial H}{\partial
p_i}-\frac{\partial F}{\partial p_i}\frac{\partial H}{\partial
q_i}\big)=0.
$$

The left hand side of \eqref{horizontal2} is
$$
\sum_i(c^i dq_i- b^i dp_i) + a\big(dH-\frac{\partial H}{\partial t}dt\big)-\sum_j\big(c^j\frac{\partial H}{\partial p_j}+b^j\frac{\partial H}{\partial q_j}\big)dt.
$$

Therefore, by comparing the terms with $dp_i$, $dq_i$, and $dt$ in \eqref{horizontal2}, respectively, we obtain:
\begin{align}
\label{dp}  b^i-a \frac{\partial H}{\partial p_i}  = \frac{\partial F}{\partial p_i},   &\\
\label{dq}  c^i+a \frac{\partial H}{\partial q_i}   = -\frac{\partial F}{\partial q_i},&\qquad i=1,\dots,n,\\
\label{dt}  \sum_j\big(c^j\frac{\partial H}{\partial p_j}+b^j\frac{\partial H}{\partial q_j}\big)  =  \frac{\partial F}{\partial t}.&
\end{align}

By multiplying \eqref{dp} with $p_i$, and taking the sum of all $i$, from \eqref{horizontal}, we get
\begin{equation}
\label{a} a = -\rho^{-1}\sum_j\frac{\partial F}{\partial p_j}p_j.
\end{equation}
Next, by substitution of \eqref{a} into \eqref{dp} and \eqref{dq}, we get, respectively:
\begin{align}
\label{b} b^i = -\rho^{-1}\sum_j\frac{\partial F}{\partial p_j}p_j\frac{\partial H}{\partial p_i}+\frac{\partial F}{\partial p_i},&\\
\label{c} c^i = \rho^{-1}\sum_j\frac{\partial F}{\partial p_j}p_j\frac{\partial H}{\partial q_i}-\frac{\partial F}{\partial q_i},& \qquad  i=1,\dots,n.
\end{align}
Now, the equation \eqref{dt} is equivalent to the property that $F$ is the integral of the canonical equations \eqref{1}.
\end{proof}

The inverse Noether theorem for symmetries of $k$-th order Lagrangians is given recently in \cite{FS}. The set $U_H$ corresponds to the set $\mathcal W$ given there.

Note that in many well studied examples of natural mechanical systems, such us Kovalevskaya top (see \cite{D, SS}\footnote{There, one can find the formulation of the Noether theorem in quasi-coordinates within Lagrangian setting, such that transformations of time and coordinates $(t,q)$  depend on $(t,q,\dot q)$. Recently, this approach is extended to nonconservative systems in \cite{M} as well.}), the integrals are interpreted as Noether
integrals with the gauge terms. Here we have the following statement.

\begin{cor}\label{posledica}
Consider a natural mechanical system on $T^*Q$ with the Hamiltonian of the form $H=T+V$, where $T=\frac12\sum_{ij} g^{ij}(t,q)p_ip_j$ is the positive definite kinetic energy and $V=V(t,q)$ is the potential. If the potential is bounded from the above, $\max_{(t,q)\in \R\times Q} V(t,q)<v$, then we can take the same system with the potential replaced by $V-c$, where $c>v$. Then every integral $F$ of the Hamiltonian equations \eqref{1} is a Noether integral with the Noether symmetry $\zeta=Y_F$.
\end{cor}

\begin{proof}
The elementary action $\rho=T-V$ is always greater then 0. Therefore \eqref{cm} is a strictly contact manifold.
\end{proof}

\begin{rem}
If there exist a closed 1-form $\beta$,
$$
\beta=a dt+ \sum_i b_i dq_i+c_i dp_i,
$$
such that
\begin{equation}\label{kontakni uslov2}
\rho+a+\sum_i  b_i\frac{\partial H}{\partial p_i}-c_i\frac{\partial H}{\partial q_i}\ne 0, \quad \text{for all}\quad (t,q,p),
\end{equation}
then the extended phase space $\R\times T^*Q$ is a strongly contact manifold with respect to the contact 1-form $pdq-Hdt+\beta$.
To every integral $F$ of the Hamiltonian equation \eqref{1}, we can associate unique weak Noether symmetry $\zeta$,
the contact Hamiltonian flow of the integral $F$ with repect to the contact form $pdq-Hdt+\beta$,
such that the corresponding Noether integral is equal to $F$: $i_\zeta(pdq-Hdt+\beta)=F$.
For example, the replacement of potential energy $V$ by $V-c$ in the Corollary \ref{posledica} corresponds to the closed form $\beta=a dt$, $a\equiv c$.

\end{rem}

\begin{prop}\label{jaka-slaba}
Assume that $\zeta$ is a weak Noether symmetry that satisfies \eqref{wns}. Then, the associated Noether symmetry $\tilde\zeta$ on $U_H$ with the same
conserved quantity is
$$
\tilde\zeta=\zeta+\rho^{-1}(\beta(\zeta)-f)Z,
$$
where $Z$ is given by \eqref{reb1}.
\end{prop}

\begin{proof} According to Theorem \ref{treca}, on $U_H$ we have the Noether symmetry $\tilde\zeta$, such that
\begin{equation}\label{ji}
J=i_{\zeta}(pdq-Hdt+\beta)-f=i_{\tilde\zeta}(pdq-Hdt).
\end{equation}

From \eqref{izvod2}, \eqref{ji} we have
$$
dJ=-i_{\zeta}(dp\wedge dq-dH\wedge dt)=-i_{\tilde\zeta}(dp\wedge dq-dH\wedge dt),
$$
and, therefore,
$$
\zeta-\tilde\zeta\in\ker(dp\wedge dq-dH\wedge dt).
$$

Thus,
$$
\tilde\zeta=\zeta+\nu Z,
$$
for some function $\nu(t,q,p)$. Finally, by substitution of the
above relation to \eqref{ji}, we get
$\nu=\rho^{-1}(\beta(\zeta)-f)$.
\end{proof}

\begin{exm}\label{LH}\emph{Linear integrals and energy.}
Assume that
$$
F=\sum_j\xi^j(q,t)p_j
$$
is the first integral. Then
$$
\sum_j\frac{\partial F}{\partial p_j}p_j=F,
$$
and Theorem \ref{treca} gives the well known expression for the Noether symmetry
\begin{equation*}
\zeta=\sum_i\xi^i\frac{\partial}{\partial q_i}-\sum_{ij}\frac{\partial\xi^j}{\partial q_i}p_j\frac{\partial}{\partial p_i}.
\end{equation*}

Next, if the Hamiltonian $H$ does not depend on time it is the integral of the system. The Noether symmetry from Theorem \ref{treca}, for the integral
$F=-H$,
takes the expected form:
$$
\zeta=\frac{\partial
}{\partial t}.
$$

Note that the above vector fields have smooth extensions from $U_H$ to $\R\times T^*Q$.
\end{exm}

\begin{exm}\emph{Quadratic integrals of the geodesic flows.}
Consider the geodesic flow with the Hamiltonian function $H=\frac12\sum_{ij}g^{ij}(q)p_ip_j$. We have $\rho=H\ne 0$ outside the zero section of $T^*Q$.
Assume that we have a quadratic first integral
$$
F=\frac12\sum_{ij} a^{ij}(q)p_ip_j.
$$
Then $$
\sum_j\frac{\partial F}{\partial p_j}p_j=2F
$$
and outside the zero section of $T^*Q$ we have the Noether symmetry
\begin{equation*}
\zeta=-\frac{F}{H}\frac{\partial
}{\partial t}+\sum_i\big(-\frac{F}{H}\frac{\partial H}{\partial p_i}+\frac{\partial F}{\partial p_i}\big)\frac{\partial}{\partial
q_i}+\sum_i\big(\frac{F}{H} \frac{\partial H}{\partial q_i}-\frac{\partial F}{\partial q_i}\big)\frac{\partial}{\partial p_i}.
\end{equation*}
\end{exm}

\begin{exm}\label{Kepler}\emph{Kepler problem.}
The Noether symmetries
associated to the Runge--Lenz vector in the Kepler problem are one of the basic examples for the inverse Noether theorem, see \cite{CMN, SC1, GS2, J}.
Consider a planar motion of a unit mass particle in the central gravitational force field. We have $Q=\R^2\setminus\{(0,0)\}$, and the Hamiltonian is
$$
H=T+V=\frac12(p_1^2+p_2^2)-\frac{\mu}{r}, \qquad r=\sqrt{q_1^2+q_2^2}, \qquad \mu>0
$$
The system is superintegrable with well known integrals: the
Hamiltonian $H$, the angular momentum $L=q_1p_2-q_2p_1$, and the
Runge--Lenz vector
$$
A=(A_1,A_2)=\big(q_1p_2^2-q_2p_1p_2-\mu\frac{q_1}{r},q_2p_1^2-q_1p_1p_2-\mu\frac{q_2}{r}\big).
$$

The elementary action $\rho=T-V$ is greater then 0 and  the Nether symmetries for the integrals $H$ and $L$ are already described in Example \ref{LH}.
We have
$$
\frac{\partial A_k}{\partial p_j}p_j=2A_k+2\mu\frac{q_k}{r}, \qquad k=1,2.
$$

Therefore, the Noether symmetries of integrals $A_1$ and $A_2$ are
\begin{equation*}
\zeta_k=\tau_k(t,q,p)\frac{\partial
}{\partial t}+\xi^i_k(t,q,p)\frac{\partial}{\partial
q^i}+\eta^i_k(t,q,p)\frac{\partial}{\partial p_i},\qquad k=1,2,
\end{equation*}
where the coefficient $\tau_k,\xi^i_k,\eta^i_k$, $k=1,2$,
are given by
\begin{align*}
\tau_1 &= -\rho^{-1}\big(q_1p_2^2-q_2p_1p_2+\mu\frac{q_1}{r}\big), \\
\xi^1_1 &= {\tau_1}p_1-q_2p_2,\\
\xi^2_1 &= {\tau_1}p_2+2q_1p_2-q_2p_1,\\
\eta^1_1 &= -\tau_1\mu\frac{q_1}{r^3}-p_2^2-\mu\frac{q^2_1}{r^3}+\mu\frac{1}{r},\\
\eta^2_1 &= -\tau_1\mu\frac{q_2}{r^3}+p_1p_2-\mu\frac{q_1q_2}{r^3},
\end{align*}
\begin{align*}
\tau_2 &= -\rho^{-1}\big(q_2p_1^2-q_1p_1p_2+\mu\frac{q_2}{r}\big), \\
\xi^1_2 &= {\tau_2}p_1+2q_2p_1-q_1p_2,\\
\xi^2_2 &= {\tau_2}p_2-q_1p_1,\\
\eta^1_2 &= -\tau_2\mu\frac{q_1}{r^3}+p_1p_2-\mu\frac{q_1q_2}{r^3},\\
\eta^2_2 &= -\tau_2\mu\frac{q_2}{r^3}-p_1^2-\mu\frac{q_2^2}{r^3}+\mu\frac{1}{r}.
\end{align*}
\end{exm}

\section{Integrability by means of Noether symmetries}

For $\rho\ne 0$, the Noether symmetries are contact vector fields
and we can use the notion of complete integrability of contact
systems (see \cite{KT, Jo3, Jo2}) to obtain a variant of the
complete (non-commutative) integrability in terms of Noether
symmetries of time-dependent Hamiltonian systems. It appears,
however, that we do not need the contact assumption $\rho\ne 0$.

\begin{thm}\label{1.i}
Assume that Hamiltonian equations \eqref{1} have $m$ independent
Noether symmetries $\zeta_1,\dots,\zeta_m$, independent of the
vector field \eqref{reb1}, such that first $r$ of then commute
with all symmetries,
$$
[\zeta_i,\zeta_j]=0, \quad i=1,\dots,r, \quad j=1,\dots,m,
$$
and $2n=m+r$. Then

{\rm (i)} The Noether integrals $J_i=i_{\zeta_i}(pdq-Hdt)$ are
independent and the equations \eqref{1} are locally solvable by
quadratures.

{\rm (ii)} If the vector fields $Z,\zeta_1,\dots,\zeta_r$ are
complete, then a connected regular component of the invariant
variety in the extended phase space
\begin{equation}\label{inv povrs}
M_c=\{(t,q,p)\in \R \times T^*Q \,\vert\, J_i=c_i, \,  i=1,\dots,m\}
\end{equation}
is diffeomorphic to a cylinder $\mathbb T^l \times \R^{r+1-l}$,
for some $l$, $0\le l \le r$, where $\mathbb T^l$ is a
$l$--dimensional torus. There exist coordinates
$\varphi_1,\dots,\varphi_l,x_1,\dots,x_{r+1-l}$ of $\mathbb T^l
\times \R^{r+1-l}$, which linearise the equation in the extended
phase space:
\begin{eqnarray*}
&&\dot\varphi_i=\omega_i=const, \qquad i=1,\dots,l,\\
&&\dot x_j=a_j=const, \qquad  j=1,\dots,r+1-l.
\end{eqnarray*}
\end{thm}

Note that the above statement slightly differs from the
Arnold--Liouville and non-commutative integrability of
time-dependent Hamiltonian systems studied in \cite{GMS, GS3}.

\begin{proof}
(i) According to the item (iii) of Theorem \ref{prva}, we have
$$
\left[Z,\zeta_i\right]=\nu_i Z,
$$
for some smooth functions $\nu_i=\nu_i(t,q,p)$, $i=1,\dots,m$. We
need to find functions $f_i$, such that the vector fields
$$
Z, \qquad \tilde\zeta_i=\zeta_i-f_i Z
$$
pairwise commute:
\begin{equation}\label{komutiranje}
[Z,\tilde\zeta_i]=0, \qquad [\tilde\zeta_i,\tilde\zeta_j]=0,
\qquad i,j=1,\dots,r.
\end{equation}

Let $\tau_i=dt(\zeta_i)$. Since $dt(Z)=1$, we have
$$
\nu_i=[Z,\zeta_i](dt)=Z(dt(\zeta_i))-\zeta_i(dt(Z))=Z(\tau_i)=\frac{\partial
\tau_i}{\partial t}+\frac{\partial \tau_i}{\partial
q}\frac{\partial H}{\partial p}-\frac{\partial \tau_i}{\partial
p}\frac{\partial H}{\partial q}.
$$
Therefore, for $f_i=\tau_i=dt(\zeta_i)$, we have
$$
[Z,\tilde\zeta_i]=[Z,\zeta_i-\tau_i Z]=Z(\tau_i)Z-Z(\tau_i)Z=0.
$$

Further,
\begin{align*}
[\tilde\zeta_i,\tilde\zeta_j] &=    [\zeta_i-\tau_i Z,\zeta_j-\tau_j Z]\\
&=  [\zeta_i,\zeta_j]-[\zeta_i,\tau_j Z]-[\tau_i Z,\zeta_j]+[\tau_i Z,\tau_j Z]\\
&=  0-(\zeta_i(\tau_j)Z-\tau_j Z(\tau_i) Z)-(\tau_i Z(\tau_j) Z-\zeta_j(\tau_i)Z)\\
&\qquad +(\tau_i Z(\tau_j)-\tau_j Z(\tau_i)Z)\\
&= (\zeta_j(\tau_i)-\zeta_i(\tau_j))Z.
\end{align*}

On the other hand, since $[\zeta_i,\zeta_j]=0$, we have
$$
[\zeta_i,\zeta_j](dt)=\zeta_i(\tau_j)-\zeta_j(\tau_i)=0,
$$
which proves \eqref{komutiranje}.

Now, consider the invariant variety \eqref{inv povrs}. At a
generic point $(t,q,p)$, the differentials of the integrals $J_i$
are independent. Indeed, in our case \eqref{izvod2} reads
\begin{equation}\label{diferencijali}
dJ_i=-i_{\zeta_i} d(pdq-Hdt), \qquad i=1,\dots,m.
\end{equation}

Therefore, if there exist real parameters $a_1,\dots,a_m$,
$a_1^2+\dots+a_m^2\ne 0$, such that
$$
a_1 dJ_1+\dots+a_m dJ_m=0,
$$
then
$$
i_{a_1\zeta_1+\dots+a_m\zeta_m}\in \ker d(pdq-Hdt),
$$
which implies that $Z,\zeta_1,\dots,\zeta_m$ are dependent at
$(t,q,p)$. Thus, the integrals $J_1,\dots,J_m$ are independent and
the regular invariant levels sets \eqref{inv povrs} are
$(r+1)$--dimensional submanifolds.

Since $J_i$ are the integrals of the equations \eqref{1}, the
vector field $Z$ is tangent to $M_c$. Further, we have
\begin{align}
0 & =i_{[\zeta_i,\zeta_j]}(pdq-Hdt)\nonumber\\
&=L_{\zeta_i}\circ i_{\zeta_j}(pdq-Hdt)-i_{\zeta_j}\circ L_{\zeta_i}(pdq-Hdt)\label{izvodi}\\
&={\zeta_i} (J_j), \qquad \qquad i=1,\dots,r,\qquad
j=1,\dots,m.\nonumber
\end{align}

Thus, the commuting vector fields $\tilde\zeta_1=\zeta_1-\tau_1
Z,\dots,\tilde\zeta_r=\zeta_r-\tau_r Z$ are also tangent to $M_c$
and, by the Lie theorem \cite{Koz}, the trajectories of \eqref{1} can be found locally by
quadratures.

\medskip

(ii) The proof of item (ii) is the same as the corresponding
statement in the Arnold--Liouville theorem (see \cite{Ar}).
\end{proof}

If the Hamiltonian and Noether symmetries are periodic with
respect to the time translation $(t,q,p)\mapsto (t+1,q,p)$,  we
can consider $S^1\times T^*Q (t,q,p)$, $S^1=\R/\mathbb Z$, as an
extended phase space.

With the above notation we have

\begin{cor}\label{2.i}
The regular compact connected components of $M_c/\mathbb Z$ are
$(r+1)$-dimensional tori with quasi-periodic dynamics
$$
\dot\varphi_i=\omega_i=const, \qquad i=1,\dots,r+1.
$$
\end{cor}

\begin{rem}
Assume that the vector fields $\zeta_i$ are weak Noether
symmetries
\begin{equation}\label{wns2}
L_{\zeta_i}(pdq-Hdt+\beta)=df_i, \qquad i=1,\dots,m.
\end{equation}
Then \eqref{komutiranje} still holds, the vector field $Z$ is
tangent to $M_c$, and \eqref{izvod2} implies \eqref{diferencijali}, where the Noether
integrals are $J_i=i_{\zeta_i}(pdq-Hdt+\beta)-f_i$, $i=1,\dots,m$.

Now,
\begin{align*}
0 & =i_{[\zeta_i,\zeta_j]}(pdq-Hdt+\beta)\nonumber\\
&=L_{\zeta_i}\circ i_{\zeta_j}(pdq-Hdt+\beta)-i_{\zeta_j}\circ L_{\zeta_i}(pdq-Hdt+\beta)\label{izvodi*}\\
&={\zeta_i}(J_j)+\zeta_i(f_j)-\zeta_j(f_i)
\end{align*}
and
\begin{align*}
0 & =i_{[\zeta_i,\zeta_j]}d(pdq-Hdt)\nonumber\\
&=L_{\zeta_i}\circ i_{\zeta_j}d(pdq-Hdt)-i_{\zeta_j}\circ L_{\zeta_i}d(pdq-Hdt)\\
&=L_{\zeta_i}(dJ_j)-i_{\zeta_j}(i_{\zeta_i}d^2(pdq-Hdt)-d^2 J_i)\\
&=d(\zeta_i(J_j)),\qquad \qquad i=1,\dots,r,\qquad
j=1,\dots,m.\nonumber
\end{align*}

Thus, we obtain $\zeta_i(J_j)=c_{ij}=const$. However, in order to
have $c_{ij}=0$, the additional assumptions
$\zeta_i(f_j)=\zeta_j(f_i)$, $i=1,\dots,r$, $j=1,\dots,m$ should
be added in Theorem \ref{druga}.
\end{rem}

\subsection*{Acknowledgments}
I am grateful to Borislav Gaji\'c for useful suggestions and comments.
The research was supported by the Serbian Ministry of
Science Project 174020, Geometry and Topology of Manifolds,
Classical Mechanics and Integrable Dynamical Systems.

\end{document}